\documentclass[final, reqno, 11pt]{amsart}
\usepackage{amsmath}
\usepackage{amsfonts}
\usepackage{amssymb}
\usepackage{amsthm}

\newtheorem{theorem}{Theorem}[section]
\newtheorem{lemma}[theorem]{Lemma}
\newtheorem{proposition}[theorem]{Proposition}
\newtheorem{corollary}[theorem]{Corollary}
\theoremstyle{remark}

\newcommand{\set}{\mathbb}
\newcommand{\R}{\mathbb R}

\newcommand{\B}{\mc L}

\newcommand{\dl}{\nabla}

\newcommand{\les}{\lesssim}
\newcommand{\mc}{\mathcal}
\newcommand{\be}{\begin{equation}}
\newcommand{\ee}{\end{equation}}
\newcommand{\bee}{\begin{align}}
\newcommand{\eee}{\end{align}}

\newcommand{\ba}{\begin{array}}
\newcommand{\ds}{\displaystyle}

\newcommand{\ea}{\end{array}}
\newcommand{\bpm}{\begin{pmatrix}}
\newcommand{\epm}{\end{pmatrix}}
\newcommand{\lb}{\label}
\DeclareMathOperator{\sgn}{sgn}

\DeclareMathOperator{\Imim}{Im}

\newcommand{\ov}{\overline}
\newcommand{\dd}{{\,}{d}}

\renewcommand{\Im}{\Imim}

\title{The Schr\"{o}dinger Equation with Potential in Random Motion}
\author{Marius Beceanu}
\address{110 Frelinghuysen Rd., Rutgers Math.\ Dept., Piscataway, NJ 08854, USA}
\email{mbeceanu@math.rutgers.edu}
\author{Avy Soffer}
\address{110 Frelinghuysen Rd., Rutgers Math.\ Dept., Piscataway, NJ 08854, USA}
\email{soffer@math.rutgers.edu}

\begin{document}
\maketitle
\numberwithin{equation}{section}
\begin{abstract}
We study Schr\"{o}dinger's equation with a potential moving along a Brownian motion path. We prove a RAGE-type theorem and Strichartz estimates for the solution on average.
\end{abstract}

\section{Introduction}
\subsection{Overview}

Consider the linear Schr\"{o}dinger equation in $\set R^3$ with a potential subject to Brownian motion
\be\lb{1.1}
i \partial_t Z + H(t) Z = 0,\ Z(0)= Z_0 \text{ given},
\ee
where
\be\lb{eq_3.2}
H(t) = H_0 + V(t) = -\Delta + V(x- \alpha \gamma(t)).
\ee
$\gamma(t)$ is a typical path of the standard Brownian motion $B_t$ and $\alpha$ is a parameter for its magnitude.

The potential $V:=V(x-\gamma(t))$ moves along a random path $\gamma(t)$ of Brownian motion $B_t$, instead of a deterministic $\dot H^{1/2} \cap C$ path as in \cite{becsof}.

Assume that $H=-\Delta+V$ has bound states. If the potential $V(x-\alpha \gamma(t))$ moves sufficiently slowly along a smooth path, by physical considerations bound states are preserved under (\ref{1.1}). Furthermore, mass and energy are conserved quantities for Schr\"{o}dinger's equation: when $\gamma(t) \equiv 0$,
$$\begin{aligned}
M(t) &:= \int_{\R^3} |Z(x, t)|^2 \dd x,\\
E(t) &:= \int_{\R^3} |\dl Z(x, t)|^2 + V(x-\gamma(t)) |Z(x, t)|^2 \dd x
\end{aligned}$$
are constant in $t$. More generally, $M(t)$ remains constant for $\gamma(t) \in B^{1/2}_{2, \infty}$ and we show in \cite{becsof} that if $E[Z(0)]$ is finite, then $E[Z(t)]$ remains bounded for all time.

Conversely, in this paper we prove that when $\gamma(t)$ is a typical Brownian motion path, the solution completely disperses. A typical path of Brownian motion is in $B^{1/2}_{2, \infty}$ locally in $t$, which differs only logarithmically from $\dot H^{1/2}\cap C$ considered in \cite{becsof}. This logarithmic difference produces qualitatively different results.

\subsection{Sample paths} For typical paths of Brownian motion, we obtain uniform control of the solution on finite time intervals in Section \ref{sect_2}. Thus we show that ground states of the equation disappear in finite time.

Without loss of generality, we consider the time interval $[0, 1]$. Let $V = V_1 V_2$, where $V_1 = |V|^{1/2}$ and $V_2 = |V|^{1/2} \sgn V$. 

Given $\gamma:[0, 1] \to \set R^3$, define the operator $S(\gamma)$ on $L^2_{[0, 1]} L^2_x$ by
$$
S(\gamma)(t, s) := \chi_{[0, 1]}(t) \int_0^t V_2 e^{i(t-s)\Delta} e^{(\gamma(t)-\gamma(s)) \dl} V_1 F(s) \dd s.
$$

\begin{proposition}\lb{long_int}
Consider $V \in L^{3/2, \infty}_0$. Then almost every sample path $\gamma$ of $B_t$, $t \in [0, \infty)$, has the property that, for any initial data $Z_0$ and any $\epsilon>0$, the corresponding solution $Z$ of (\ref{1.1}) fulfills
$$
\mu(\{t \mid \|Z(t)\|_{L^6_x} < \epsilon\}) = + \infty.
$$
\end{proposition}

Thus, almost every random path leads to the dispersion of all initial data.

Keeping better track of computations, we obtain instead a quantitative result:
\begin{proposition}\lb{finite_time} Assume $V = V_{3/2} + V_{\infty}$, where $V_{3/2} \in L^{3/2, \infty}_0$,  $V_{\infty} \in L^{\infty}_0$, meaning $V_{\infty} \in L^{\infty}$ and $\lim_{x \to \infty} |V_{\infty}(x)| = 0$.

Then, for any $\epsilon>0$
$$
\lim_{\alpha \to \infty} \set E \|S(\alpha B_t)\|_{L^2_{t, x} \to L^2_{t, x}} = 0.
$$
Additionally, fix a time $T>0$. Then
$$
\lim_{\alpha \to \infty} \set E \Big\{\sup_{\|Z_0\|_2 = 1} \|Z-e^{it\Delta} Z_0\|_{L^2_{[0, T]} L^{6, 2}_x}\Big\} =0.
$$
\end{proposition}
An analogous result holds in all spatial dimensions one and higher.

If $\gamma$ is a piecewise linear curve, then $\|S(\alpha \gamma)\|$ will never go to zero, regardless of the size of $\gamma$.

Thus, a generic path of Brownian motion is unlike a piecewise linear curve. The difference is that the graph of Brownian motion has Haussdorff dimension strictly greater than $1$. As such, the result of Proposition \ref{finite_time} will hold for other random processes, such as fractional Brownian motions with different Hurst coefficients, whose Haussdorff dimension is greater than $1$.

As a consequence of Proposition \ref{finite_time} we obtain a result analogous to the RAGE theorem.

The RAGE theorem of Ruelle \cite{ruelle}, Amrein--Georgescu \cite{amge}, and Enss \cite{ens} is a classical result that rigorously establishes the dichotomy between dispersive states, whose time average with respect to any relatively compact operator $C$ is zero
$$
\lim_{T \to \infty}\frac 1 T \int_0^T \|C \exp(-itA ) \phi\|^2 dt = 0.
$$
and bound states, whose time average can be nonzero in the above sense.

Proposition \ref{finite_time} implies the uniform decay of solutions on an infinite time interval, in a way similar to the RAGE theorem:
\begin{proposition}\lb{cor3.5}
Let $\gamma$ be a typical path of $B_t$. As $\alpha \to \infty$, for a solution $Z$ to equation (\ref{1.1}) driven by $\alpha \gamma(t)$ and of initial data $Z_0$
$$
\lim_{\alpha \to \infty} \set E \Big\{\sup_{\|Z_0\|_2 = 1} \frac 1 T {\|Z\|_{L^2_{[0, T]} L^{6, 2}_x}^2}\Big\} = 0.
$$
\end{proposition}
Thus, over almost any sample path $\gamma$, all solutions will spend a higher proportion of the total time, approaching unity, in an ionized state as the temperature increases --- uniformly for all initial data.

Proposition \ref{cor3.5} shows that that all initial data are equally affected by dispersion.

The reason why there is no further improvement over an infinite time interval is the following: for a curve $\gamma:[0, \infty) \to \set R^3$, define the operator $T(\gamma)$ on $L^2_{t, x}$ by
\be\lb{def_t}
(T(\gamma) F)(t) := \int_0^t V_2 e^{i(t-s)\Delta} e^{(\gamma(t)-\gamma(s)) \dl} V_1 F(s) \dd s.
\ee
\begin{proposition}\lb{prop3.6} Let $V \in L^{3/2, \infty}_0$. With probability one, $\|T(B_t)\|_{\B(L^2, L^2)}$ is equal to a constant, which at least equals $\|T(0)\|_{\B(L^2, L^2)}$.
\end{proposition}
This shows the difference between a finite and an infinite time interval.

\subsection{Strichartz estimates}
In Section \ref{strichartz} we prove endpoint Strichartz estimates for fixed initial data, which hold on average for all Brownian motion paths.
\begin{theorem}\lb{thm21}
Consider a solution $Z$ of equation (\ref{1.1}) on $\set R^3$
$$
i \partial_t Z + H(t) Z = 0,\ Z(0) = Z_0,\ H(t) = -\Delta + V(x-\alpha B_t).
$$
Assume that the potential $V$ is in the Lorentz space $L^{3/2, 1}$.
Then there exists $\alpha_0>0$ such that, for any $\alpha \geq \alpha_0$, Strichartz estimates hold and the solution disperses almost surely:
$$
\sup_{\|Z_0\|_2 = 1} \set E \|Z-e^{it\Delta} Z_0\|_{L^{\infty}_t L^2_x \cap L^2_t L^{6, 2}_x} \leq C \alpha^{-2} \|V\|_{L^{3/2, 1}}.
$$
\end{theorem}

A similar conclusion holds in all dimensions three or higher.

\subsection{History of the problem} Schr\"{o}dinger's equation with a potential in random motion was first studied by Pillet \cite{pillet}, who proved that wave operators are unitary with probability one. The random process considered was confined to a bounded space region and the process and the potential had to satisfy further smoothness and nondegeneracy assumptions.

Pillet's results were extended by Cheremshantsev \cite{cherem}, \cite{cherem2} to the case of Brownian motion. Cheremshantsev proved that wave operators were unitary with probability one when the potential $V$ is in $L^2$ and decays like $|x|^{-5/2}$ at infinity.

In this paper we consider two approaches. One choice is to take a typical path $\gamma$ of the Brownian motion $B_t$ and investigate the dispersive properties of (\ref{1.1}), with the potential $V(x-\gamma(t))$ moving along the path $\gamma$, for any initial data. We prove the uniform decay of solutions, averaged in time, for each sample path and all initial data --- Proposition \ref{cor3.5}.

Another approach is to examine, for each initial data $Z(0)$, the properties of the solution --- dispersion and energy boundedness --- averaged over all Brownian motion paths. We prove Strichartz estimates in this setting --- Theorem \ref{thm21}.

\subsection{Notations}

We denote Lorentz spaces by $L^{p, q}$.

$a \les b$ means that $|a| \leq C |b|$ for some constant $C$.

For an operator $T$, we denote its integral kernel by $T(t, s)$:
$$
(T F)(t) = \int_\R T(t, s) F(s) \dd s.
$$



\section{Dispersive bounds for sample paths of Brownian motion}
\lb{sect_2}

\begin{lemma}\lb{lemma_3.1}  Assume $V = V_{3/2} + V_{\infty}$, where $V_{3/2} \in L^{3/2, \infty}_0$,  $V_{\infty} \in L^{\infty}_0$, meaning that $V_{\infty} \in L^{\infty}$ and $\ds\lim_{x \to \infty} |V_{\infty}(x)| = 0$.

Then, for any continuous curve $\gamma_0$ and $\epsilon>0$, with positive probability $\|S(B_t) - S(\gamma_0)\|_{L^2_{t, x} \to L^2_{t, x}} < \epsilon$.
\end{lemma}

\begin{proof}
Consider any continuous path $\gamma_0$. For any $\delta>0$, with positive probability $\|B_t - \gamma_0\|_{L^{\infty}_{[0, 1]}} < \delta$. Furthermore, for fixed $\gamma_0$, by Lemma 3.1 of \cite{becsof}
$$
\lim_{\|\gamma - \gamma_0\|_{L^{\infty}} \to 0} \|S(\gamma) - S(\gamma_0)\|_{L^2_{t, x} \to L^2_{t, x}} = 0.
$$
\end{proof}

An immediate consequence is that, with positive probability, $S(B_t)$ is close in norm to $S(0)$. Furthermore, by dominated convergence
$$
\lim_{\alpha \to 0} \set E \|S(\alpha B_t) - S(0)\|_{L^2_{t, x} \to L^2_{t, x}} = 0.
$$

We also use this lemma to deduce that, with positive probability, the norm of $T(B_t)$ becomes arbitrarily small on compact intervals.
\begin{proposition}\lb{prop_3.2} Consider $V \in L^{3/2, \infty}_0 + L^{\infty}_0$, where $L^{3/2, \infty}_0$ is the weak-$L^{3/2}$ closure and $L^{\infty}_0$ is the $L^{\infty}$ closure of the set of smooth, compactly supported functions.

Then, for every $\epsilon>0$, $\set P\big(\|S(B_t)\|_{\B(L^2_{t, x}, L^2_{t, x})} < \epsilon\big) > 0$.
\end{proposition}
\begin{proof}
To begin with, we consider $V$ smooth and of compact support and then proceed by approximation for $V \in L^{3/2, \infty}_0 + L^{\infty}_0$.

It suffices to construct a continuous curve $\gamma_0: [0, 1] \to \set R^3$ such that $\|S(\gamma_0)\|_{\B(L^2_{t, x}, L^2_{t, x})} < \epsilon/2$. To this purpose, subdivide the interval $[0, 1]$ into $n$ equal subintervals and let
$$
S(\gamma_0) = \sum_{1 \leq k \leq j \leq n} T_{jk},\ T_{jk}(t, s) = \chi_{[\frac {j-1} n, \frac j n]}(t) S(\gamma)(t, s) \chi_{[\frac {k-1} n, \frac k n]}(s).
$$
Note that
$$
\|S(\gamma_0)\|_{\B(L^2_{t, x}, L^2_{t, x})} = \sup \{ \langle S(\gamma_0) F, G \rangle \mid \|F\|_{L^2_{t, x}} = \|G\|_{L^2_{t, x}} = 1 \}
$$
and likewise for $T_{jk}$. We obtain
$$
\|S(\gamma_0)\|^2_{L^2_{t, x} \to L^2_{t, x}} \leq \sum_{1 \leq k \leq j \leq n} \|T_{jk}\|^2.
$$
Take $\gamma_0$ to be piecewise linear, i.e.\ linear on each subinterval $\Big[\frac {j-1} n, \frac j n\Big]$:
$$
\gamma_0(t) = \gamma_0\bigg(\frac {j-1} n\bigg) + \bigg(t - \frac {j-1} n \bigg) v_j,\ t \in \bigg[\frac {j-1} n, \frac j n\bigg].
$$
By the $L^2$ boundedness of the evolution, which is unitary for (\ref{1.1}),
$$
\|T_{jk} F\|_{L^2_{t, x}} \leq \frac C n \|F\|_{L^2_{t, x}}. 
$$
On the other hand, by smoothing estimates
$$
\|T_{jk} F\|_{L^2_{t, x}} \leq \frac C {|v_j - v_k|} \|F\|_{L^2_{t, x}}.
$$
Combining the two estimates, we obtain
$$
\|T_{jk} F\|_{L^2_{t, x}} \leq \frac C {n \langle \frac{v_j - v_k}n \rangle} \|F\|_{L^2_{t, x}}.
$$
Thus
$$
\|S(\gamma_0)\|^2_{\B(L^2_{t, x}, L^2_{t, x})} \leq \frac C {n^2} \sum_{1 \leq k \leq j \leq n} \frac 1 {\langle \frac{v_j - v_k}n \rangle^2}.
$$
Furthermore, $\sum_{j=1}^n \|T_{jj}\|^2 \leq \frac C n$.

For sufficiently large $n$, the diagonal terms become arbitrarily small. The same happens with the off-diagonal terms after choosing $(v_j)_{1 \leq j \leq n}$ sufficiently far apart --- making sure that, for example, whenever $j \ne k$ one has $|v_j-v_k|>n^2$. We obtain a curve $\gamma_0$ such that
$$
\|S(\gamma_0) F\|_{L^2_{t, x}} \leq \frac {2C} n < \epsilon/2.
$$
\end{proof}

\begin{proof}[Proof of Proposition \ref{long_int}]
On the infinite time interval $[0, \infty)$, due to the Markov property of Brownian motion, any event with nonzero probability of happening in finite time will eventually occur.

In particular, ground states will be preserved on arbitrarily long time intervals.

On the other hand, on arbitrarily long intervals the random evolution will behave like the free evolution. Indeed, by Proposition \ref{prop_3.2}
$$
\|\chi_{[t_0, t_1]}(t) T(B_t) \chi_{[t_0, t_1]}(t)\|_{\mc L(L^2_{t, x}, L^2_{t, x})}
$$
can be made arbitrarily small with nonzero probability --- hence with probability one --- on some interval $[t_0, t_1]$. We obtain convergence in
$$
Z(t) = \sum_{k=0}^{\infty} \big(\chi_{[t_0, t_1]}(t) T(B_t) \chi_{[t_0, t_1]}(t)\big)^k e^{i(t-t_0)\Delta} Z(t_0)
$$
on the arbitrarily long interval $[t_0, t_1]$ on which this occurs.

Then all solutions $Z$ will satisfy endpoint Strichartz estimates uniformly and regardless of the length of the interval: $\|Z\|_{L^2_t(t_0, t_1) L^6_x} \leq C \|Z(t_0)\|_2$.

As $\|Z(t_0)\|_2 = \|Z(0)\|_2$, this implies the conclusion of Proposition \ref{long_int}.
\end{proof}

\begin{proof}[Proof of Proposition \ref{finite_time}]
Begin by assuming that $V$ is bounded and of compact support. Once we establish the conclusion in this case, it is also implied in the general case by approximating $V$ in the $L^{3/2, \infty} + L^{\infty}$ norm.

We perform a pseudoconformal transformation centered at $(0, 0)$. This takes the solution $\psi(x, t)$ of the homogenous free Schr\"{o}dinger equation and turns it into a new solution $\phi$ such that
\be
\phi(x, t) = \frac 1 {\big(i t\big)^{3/2}} \ov \psi\big(\frac {x} {t}, \frac 1 t \big) e^{i\frac {|x|^2}{2t}},
\ee
with the convention that
\be
\phi(x, 0) = (2\pi)^{-3/2} \widehat \psi(x, 0).
\ee

In particular, if $\langle x \rangle \psi(x, 0) \in L^2_x$, then $\phi(x, 0) \in H^1_x$: by Plancherel's identity
$$
\|\phi(x, 0)\|_{H^1_x} = \|\langle x \rangle \psi(x, 0)\|_{L^2_x}.
$$
The $H^1$ norm is preserved by the free flow, hence
\be
\|\phi\|_{L^{\infty}_t H^1_x} \leq C \|\psi(x, 0)\|_{\langle x \rangle^{-1} L^2}.
\ee
Let $\chi_R(x) = \chi(x/R)$ be a smooth cutoff at a distance $R$ from the origin. For the original function $\psi$, we obtain
\be\begin{aligned}
\|\chi_R(x) \psi(x, t)\|_{\dot H^1_x} &= |t|^{-3/2} \big\|\chi(x) \phi(\frac x t, \frac 1 t) e^{-\frac {|x|^2}{2t}}\big\|_{\dot H^1_x} \\
&\leq |t|^{-3/2} \|\phi(\frac x t, \frac 1 t)\|_{\dot H^1_x} + C R |t|^{-5/2} \|\chi_R(x) \phi(\frac x t, \frac 1 t)\|_{L^2} \\
&\les \frac {\langle R \rangle} {|t|} \|\phi\|_{L^{\infty}_t H^1_x} \\
&\les \frac {\langle R \rangle} {|t|} \|\psi(x, 0)\|_{\langle x \rangle^{-1} L^2}.
\end{aligned}\ee

We turn this into a bound on the mass current. For the fixed smooth compactly supported cutoff function $\chi_{R/2}$, centered around zero,
\be\begin{aligned}
\partial_t \int \chi_{R/2}(x) |\psi(x, t)|^2 \dd x = \int \dl \chi_{R/2}(x) \Im (\ov \psi(x, t) \dl \psi(x, t)) \dd x.
\end{aligned}\ee
The reader is referred to \cite{taobook} for the definitions of the pseudoconformal transformation and of the mass current.

Then,
\be\begin{aligned}
&\Big|\partial_t \int \chi_{R/2}(x) |\psi(x, t)|^2 \dd x\Big| \leq \\
&\leq C \|\chi_R(x) \psi(x, t)\|_{H^{1}} \|\chi_R(x)\psi(x, t)\|_{L^2} \\
& \leq C \min\Big(\frac {\langle R \rangle^3}{|t|^2}, \frac {\langle R \rangle^2}{|t|}\Big) \|\psi(x, 0)\|_{\langle x\rangle^{-1} L^2}^2.
\end{aligned}\ee

Let $\gamma$ be a typical path of the Brownian motion. In order to apply this inequality to it and to our problem, we perform the Galilean change of coordinates
\be
(x, t) \mapsto (x- \gamma(t_0)-v(t-t_0), t- t_0).
\ee
We replace, that is, $(t_0, \gamma(t_0))$ by $(0, 0)$, where $t_0 \in [0, 1]$ and $\gamma (t_0)$ is the position of the path at time $t_0$. In addition, we change to a coordinate system that moves at velocity $v$. $\psi(x, t)$ becomes
\be
e^{ivx} \psi(x-\gamma(t_0)-v(t-t_0), t-t_0).
\ee
For convenience, denote
\be
x_0 = \gamma(t_0),\ y_v(t) = \gamma(t_0) + v(t-t_0).
\ee
We arrive at
\be\begin{aligned}\lb{disp}
\|\chi_R(x-y_v(t)) e^{ivx}\psi(x, t)\|_{L^{\infty}_t H^{1}_x} \leq \\
\leq C \min\Big(\frac {\langle R \rangle^3}{|t-t_0|^2}, \frac {\langle R \rangle}{|t-t_0|}\Big) \|\psi(x-x_0, t_0)\|_{\langle x \rangle^{-1} L^2}.
\end{aligned}\ee

On a small time interval, the $\frac 1 {|t-t_0|}$ bound is stronger. Let $R$ be fixed such that $\chi_{R/2}$ covers the support of $V_2$ and take
$$
\psi(x, t) = e^{i(t-t_0)\Delta} V_1(x-x_0) F(x, t_0)
$$
a solution of the free Schr\"{o}dinger equation, with given data at $t_0$. We then replace $\|\psi(x-x_0, t_0)\|_{|x|^{-1} L^2}$ by $C\|F(x, t_0)\|_{L^2_x}$. Integrating (\ref{disp}) in time, we obtain
$$\begin{aligned}
\Big|\int \chi_{R/2}(x-y_v(t_1)) |\psi(x, t_1)|^2 \dd x - \int \chi_{R/2}(x-y_v(t_2)) |\psi(x, t_2)|^2 \dd x\Big| \leq \\
\leq C \big(\ln(t_2-t_0) - \ln(t_1-t_0)\big) \|F(t_0)\|_{2}^2.
\end{aligned}$$
Fix $\delta>0$. We divide the interval $[t_0+\delta, t_0+1]$ into $n$ subintervals \be\begin{aligned}
I_1 = [t_0+\delta, t_0+\delta^{(n-1)/n}],\ I_2 = [t_0+\delta^{(n-1)/n}, t_0+\delta^{(n-2)/n}],\ \ldots,\\
I_n = [t_0+\delta^{1/n}, t_0+1],
\end{aligned}\ee
of length no greater than $1-\delta^{1/n}$. For $t_1$ and $t_2$ in the same subinterval,
\be\begin{aligned}
\Big|\int V_2(x-y_v(t_1)) |\psi(x, t_1)|^2 \dd x - \int V_2(x-y_v(t_2)) |\psi(x, t_2)|^2 \dd x\Big| \leq \\
\leq \frac {C|\ln \delta|} n \|F(t_0)\|_{2}^2.
\end{aligned}\ee
If for some velocity $v$ and time $t_1 \in I_k$ we have that
$$
\Big|\int V_2(x-y_v(t_1)) |\psi(x, t_1)|^2 \dd x\Big| \geq \frac {2C|\ln\delta|} n \|F(t_0)\|_2^2,
$$
then, for all other times $t_2 \in I_k$ in the same subinterval and for the same $v$,
$$
\Big|\int V_2(x-y_v(t_2)) |\psi(x, t_2)|^2 \dd x\Big| \geq \frac {C|\ln \delta|} n \|F(t_0)\|_2^2.
$$
Consider the cylinders
$$
C_v(x_0, t_0, I_\ell, R) = \{(x, t) \mid t \in I_\ell, |x-x_0-(t-t_0)v| \leq R\}.
$$
Fix $\ell$, $1 \leq \ell \leq n$. Suppose that there exist $m$ cylinders $C_{v_k}(x_0, t_0, I_{\ell}, R/2)$ so that
\be
|v_j-v_k| \geq R \delta^{-(\ell-1)/n}
\ee
for all $1 \leq j, k \leq m$ and for each disjoint cylinder there exists $t_0 \in I_{\ell}$ such that $C_{v_k}(x_0, t_0, I_{\ell}, R/2)$ contains at least $\frac {2C |\ln \delta|} n \|F(t_0)\|_2^2$ mass at time $t_0$:
$$
\|V_2(x-y_{v_k}(t_0)) \psi(x, t_0)\|_{L^2_x}^2 \geq \frac {2C |\ln \delta|} n \|F(t_0)\|_2^2.
$$
Then each cylinder will contain at least $\frac {C |\ln \delta|} n \|F(t_0)\|_2^2$ mass at all times $t \in I_{\ell}$:
$$
\|V_2(x-y_{v_k}(t)) \psi(x, t)\|_{L^2_x}^2 \geq \frac {2C |\ln \delta|} n \|F(t)\|_2^2.
$$
The number $m$ of such disjoint cylinders containing a significant proportion of the total mass is bounded:
$$
m \leq \frac {C n}{|\ln \delta|}.
$$
Thus, for any $v$ such that $|v-v_j| \geq R \delta^{-(\ell-1)/n}$ for all $j$, $1 \leq j \leq m$, we have that for all $t \in I_{\ell}$
\be\lb{formula}
\|V_2(x-y_{v}(t)) \psi(x, t)\|_{L^2_x}^2 \leq \frac {2C |\ln \delta|} n \|F(t_0)\|_2^2.
\ee
(\ref{formula}) states that any cylinder over the interval $I_\ell$ other than the $m$ already prescribed will contain little $L^2$ mass, resulting in a small $L^2$ norm pointwise in time.

As long as $\gamma$ is not concentrated along straight line segments, we obtain a nontrivial bound for the norm of $S(\gamma)$, because $\gamma$ will spend little time in the prescribed cylinders in which the $L^2$ mass of the solution is present.

For our purpose, a good measure of concentration of $\gamma$ within cylinders, on the scale $\epsilon$, is the quantity
$$\begin{aligned}
K(\gamma, \epsilon, r) = &\sup_{I, y_0, v} \big\{m\big(\{t \in I \big| |\gamma(t) - y_0 - vt| \leq r\}\big) \mid \\
& I \subset [0, 1],\ |I| = \epsilon,\ y_0, v \in \set R^3 \big\} .
\end{aligned}$$
Here $m$ is the Lebesgue measure. Trivially $K(\gamma, \epsilon, r) \leq \epsilon$.

The operator norm is then bounded by
$$
\|S(\gamma)\|_{L^2_{t, x} \to L^2_{t, x}}^2 \leq C \Big(\delta + \frac {2 |\ln \delta|} n  + 
K(\gamma, 1-\delta^{1/n}, R)\Big).
$$

The contribution of $[t_0, t_0+\delta]$ is bounded by $C\delta$. The second term stems from times when $\gamma(t)$ is in cylinders that satisfy (\ref{formula}). The third term corresponds to $\gamma(t)$ in exceptional cylinders with large $L^2$ mass. The norm is bounded by the total length of these cylinders.

In the limit where $\frac {|\ln \delta|} n \to 0$, one has $|1-\delta^{1/n}| \leq C \frac {|\ln \delta|} n$. If we show that, for bounded $\delta$ and fixed $R$,
$$
\lim_{\alpha \to \infty} \set E K(\alpha B_t, \delta, R) = 0,
$$
then the conclusion is achieved. By rescaling, it suffices to prove this for $\delta = 1$, over the interval $I=[0, 1]$, for $B_t$ and $R=1/\alpha$:
\be\begin{aligned}\lb{prop}
&\lim_{\alpha \to \infty} \set E K(B_t, 1, 1/\alpha) = 0.
\end{aligned}\ee
Define the local time, see \cite{brown}, as the measure $\mu$ on $\set R^4$ such that
$$
\mu(A) = m(\{t \mid t \in [0, 1],\ (t, B_t) \in A\}).
$$
For $\beta < 3/2$, it is known that
$$
\set E \int \frac {\dd\mu(x) \dd\mu(y)}{|x-y|^{\beta}} < \infty.
$$
Fix $\beta$ and consider a curve $\gamma(t)$ for which this potential is finite.  In particular, then, if $B(x, \epsilon)$ is a sphere of radius $\epsilon$,
$$
\mu(B(x, \epsilon)) \leq C \epsilon^{\beta}.
$$
Assume that $|\gamma(t)|$ is bounded (true with probability close to one, if the bound is sufficiently large) and discard $\gamma([0, \epsilon])$ for some $\epsilon>0$.

To prove (\ref{prop}), we then only have to measure the presence of $\gamma$ inside slanted cylinders of length $1$ and bounded inclination. Such a cylinder of radius $\epsilon$ can be covered with $C \epsilon^{-1}$ spheres of radius $\epsilon$, meaning its measure is at most $C\epsilon^{\beta-1}$. Hence the measure goes to zero as $\epsilon \to 0$.
\end{proof}

\begin{proof}[Proof of Proposition \ref{prop3.6}]
Assume that $V \ne 0$; if $V = 0$, then the conclusion is immediate.

Fix a curve $\gamma_0:[0, R] \to \set R^3$, such that
$$
\|\chi_{[0, R]}(t) T(\gamma) \chi_{[0, R]}(t)\|_{\mc L(L^2_{t, x}, L^2_{t, x})} \ne 0,
$$
and let $\epsilon>0$. On each interval $[kR, (k+1)R]$, with positive probability,
$$\begin{aligned}
&\|\chi_{[kR, (k+1)R]}(t) T(B_t) \chi_{[kR, (k+1)R]}(t)\|_{\mc L(L^2_{t, x}, L^2_{t, x})} \geq \\
&\geq \|\chi_{[0, R]}(t) T(\gamma_0) \chi_{[0, R]}(t)\|_{\mc L(L^2_{t, x}, L^2_{t, x})} - \epsilon.
\end{aligned}$$
All these operator norms on different intervals are independent and $\|T(B_t)\|_{\mc L(L^2_{t, x}, L^2_{t, x})}$ is at least as large:
$$
\|T(B_t)\|_{\mc L(L^2_{t, x}, L^2_{t, x})} \geq \|\chi_{[kR, (k+1)R]}(t) T(B_t) \chi_{[kR, (k+1)R]}(t)\|_{\mc L(L^2_{t, x}, L^2_{t, x})}.
$$
Therefore, with probability one,
$$
\|T(B_t)\|_{\mc L(L^2_{t, x}, L^2_{t, x})} \geq \|\chi_{[0, R]}(t) T(\gamma_0) \chi_{[0, R]}(t)\|_{\mc L(L^2_{t, x}, L^2_{t, x})}-\epsilon.
$$
Since this holds for any $\epsilon>0$, also with probability one
$$\begin{aligned}
&\|T(B_t)\|_{\mc L(L^2_{t, x}, L^2_{t, x})} \geq \|\chi_{[0, R]}(t) T(\gamma_0) \chi_{[0, R]}(t)\|_{\mc L(L^2_{t, x}, L^2_{t, x})}.
\end{aligned}$$

Next, consider a curve $\gamma:[0, \infty) \to \set R^3$ such that $\|T(\gamma)\| \ne 0$ and fix $\epsilon>0$. Then there exist $F, G \in L^2_{t, x}$, of norm $1$, such that
$$
\langle T(\gamma) F, G \rangle \geq \|T(\gamma)\| - \epsilon.
$$
Let $T_R(\gamma)$ be the restriction of $T(\gamma)$ to $[0, R]$. Then for sufficiently large $R$
$$
\langle T_R(\gamma) F_R/\|F_R\|, G_R/\|G_R\| \rangle \geq (\|T(\gamma)\| - 2\epsilon)/(\|F_R\| \|G_R\|) \geq \|T(\gamma)\| - 2\epsilon.
$$
It follows that $\|T_R(\gamma)\| \to \|T(\gamma)\|$ as $R \to \infty$. Thus, with probability one,
$$
\|T(B_t)\| \geq \|T(\gamma)\|.
$$
Consider $a \in \set R$ such that $\set P(\|T(B_t)\| \geq a) > 0$. Every set of paths of positive probability contains a continuous curve, call it $\gamma$. Then with probability one $\|T(B_t)\| \geq \|T(\gamma)\| \geq a$. Take
\be
a_0 = \sup\{a \mid \set P(\|T(B_t)\| \geq a) > 0\}.
\ee
$a_0$ is finite due to the endpoint Strichartz inequality of Keel--Tao \cite{tao}. Thus $\|T(B_t)\| \geq a_0$ with probability one and, for any $\epsilon>0$, $\|T(B_t)\| \geq a_0+\epsilon$ has probability zero. Therefore $\|T(B_t)\| = a_0$ with probability one.
\end{proof}

\section{Strichartz estimates for fixed initial data}\lb{strichartz}




\begin{proof}[Proof of Theorem \ref{thm21}]
Write the evolution as a sum of terms, by means of the Duhamel formula:
\be\begin{aligned}\lb{2.2}
Z(t) &= e^{i t \Delta} Z(0) + \\
&+ \int_{t>s_1>0} e^{i(t-s_1)\Delta} V(x-\alpha B_{s_1}) e^{i s_1 \Delta} Z(0) \dd s_1 + \ldots + \\
&+\int_{t>s_1>\ldots>s_n} e^{i(t-s_1)\Delta} V(x-\alpha B_{s_1}) e^{i(s_1-s_2) \Delta} V(x-\alpha B_{s_2}) \ldots \\
&\ldots e^{i s_n \Delta} Z(0) \dd s_1 \ldots \dd s_n + \ldots .
\end{aligned}\ee
In the sequel we prove that this formal series is dominated in norm, on average, by a converging sequence:
\be\begin{aligned}
&\sum_{n=0}^{\infty} \set E \bigg\{\Big\|\int_{t>s_1>\ldots>s_n} e^{i(t-s_1)\Delta} V(x-\alpha B_{s_1}) e^{i(s_1-s_2) \Delta} V(x-\alpha B_{s_2}) \ldots \\
&\ldots e^{i s_n \Delta} Z(0) \dd s_1 \ldots \dd s_n\Big\|_{L^{\infty}_t L^2_x \cap L^2_t L^{6, 2}_x}\bigg\}< \infty.
\end{aligned}\ee
If this were the case, we could interchange the sum and the expectation and then, for every path for which the sum converged, one would obtain a bound on the Schr\"{o}dinger evolution.

We evaluate each term separately. The first is simply the free evolution and it fulfills Strichartz estimates by the result of Keel--Tao \cite{tao}. We have to show that all the remaining terms go to zero in average, as $\alpha$ goes to infinity.

The second term is also trivially bound from above by a constant, due to the endpoint Strichartz estimates. However, as $\alpha$ goes to infinity, the expected value of its norm goes to zero.

To prove this, it suffices to show that
$$\begin{aligned}
\lim_{\alpha \to \infty} \set E \big\|V(x-\alpha B_{s_1}) e^{i s_1 \Delta} Z(0) \big\|_{L^2_{s_1} L^{6/5, 2}_x} = 0.
\end{aligned}$$
In turn, this is implied by
$$
\lim_{\alpha \to \infty} \set E \big\| |V|^{1/2} (x-\alpha B_{s_1}) e^{i s_1 \Delta} Z(0) \big\|_{L^2_{s_1, x}} = 0
$$
or, equivalently,
$$
\lim_{\alpha \to \infty} \set E \int (e^{\alpha B_{s_1} \dl}|V|) \big|e^{i s_1 \Delta} Z(0) \big|^2 \dd x \dd s_1 = 0.
$$
Interchanging the order of integration, by Fubini's theorem, we compute $\set E |V|(x-\alpha B_{s_1})$ first, which equals
$$
\set E e^{\alpha B_{s_1} \dl} |V| = e^{\alpha^2 s_1 \Delta} |V| = (2 \pi \alpha^2 s_1)^{-3/2} \int e^{-|x-y|^2/(4\alpha^2 s_1)} |V|(y) \dd y.
$$
Indeed, on the Fourier side $\set E e^{\alpha B_{s_1} \xi} = e^{-\alpha^2 s_1 |\xi|^2} |V|$. More generally, since the infinitesimal generator of Brownian motion is $-\Delta$, by Dynkin's formula $\set E f(x-B_t)$ is the solution of the heat equation. We arrive at
$$
\lim_{\alpha \to \infty} \int \big(e^{\alpha^2 s_1 \Delta} |V|\big) \big|e^{i s_1 \Delta} Z(0) \big|^2 \dd x \dd s_1 = 0.
$$
Assume $V \in L^1 \cap L^{\infty}$; then
$$
\|e^{\alpha^2 s_1 \Delta} |V|\|_{\infty} \leq C \min((\alpha^2 s_1)^{-3/2}, 1).
$$
Then
$$
\int \big(e^{\alpha^2 s_1 \Delta} |V|\big) \big|e^{i s_1 \Delta} Z(0) \big|^2 \dd x \dd s_1 \leq C \int_0^{\infty} \|\big(e^{\alpha^2 s_1 \Delta} |V|\big)\|_{\infty} \dd s_1 \leq C \alpha^{-2}.
$$
In general, given that endpoint Strichartz estimates hold, it suffices to bound
$$
\set E \bigg\|\int_{s_1>\ldots>s_n} |V|^{1/2}(x-\alpha B_{s_1}) e^{i(s_1-s_2) \Delta} V(x-\alpha B_{s_2}) \ldots e^{i s_n \Delta} Z(0) \dd s_2 \ldots \dd s_n\bigg\|_{L^2_{t, x}}.
$$

To look at a concrete case, we consider the third term, whose average is bounded by the quantity
$$
\set E \int (e^{\alpha B_{s_1} \dl} |V|) \bigg| \int_{s_1>s_2} e^{i (s_1-s_2) \Delta} \big((e^{\alpha B_{s_1}} V) e^{i s_2 \Delta} Z(0)\big) \dd s_2 \bigg|^2 \dd x \dd s_1.
$$
We write the absolute value as the product between the factor and its conjugate, arriving at
\be\begin{aligned}\lb{2.11}
\set E \int_{\substack{s_1>s_2 \\ s_1>s_3}} \Big\langle (e^{\alpha B_{s_1}} |V|) e^{i (s_1-s_2) \Delta} \big((e^{\alpha B_{s_2} \dl} V) e^{i s_2 \Delta} Z(0)\big), \\
e^{i (s_1-s_3) \Delta} V(x-\alpha B_{s_3}) e^{i s_3 \Delta} Z(0) \Big\rangle \dd s_1 \dd s_2 \dd s_3.
\end{aligned}\ee
Subdivide the integral into two parts, according to whether $s_2>s_3$ or not; without loss of generality, consider the first case. Then, interchange the expectation with the integral and compute the expectation first, with respect to the $\sigma$-algebra that corresponds to time $s_2$. We obtain
\be\begin{aligned}\lb{2.12}
\set E \int_{\substack{s_1>s_2>s_3}} \Big\langle \big(e^{\alpha^2(s_1-s_2) \Delta} e^{\alpha B_{s_2} \dl} |V| \big) e^{i (s_1-s_2) \Delta} \big((e^{\alpha B_{s_2} \dl} V) e^{i s_2 \Delta} Z(0)\big), \\
e^{i (s_1-s_3) \Delta} \big((e^{\alpha B_{s_3} \dl}V) e^{i s_3 \Delta} Z(0)\big) \Big\rangle \dd s_1 \dd s_2 \dd s_3.
\end{aligned}\ee
In order to estimate this expression, we start with its building blocks. For $f$, $g \in L^2$ and a fixed potential $V$, consider the sesquilinear expression
$$
W_{|V|}(f, g) = \int_0^{\infty} \langle (e^{\alpha^2 t \Delta} |V|) e^{it\Delta} f, e^{it\Delta} g \rangle \dd t.
$$
Explicitly, the symbol of $W_{|V|}$ is given by
\be\lb{2.14}
\widehat W_{|V|}(\xi_1, \xi_2) = \frac {\widehat {|V|}(\xi_2 - \xi_1)}{\alpha^2 |\xi_1 - \xi_2|^2 - i (|\xi_1|^2 - |\xi_2|^2)}.
\ee
Using this notation, (\ref{2.12}) becomes
$$\begin{aligned}
\set E \int_{\substack{s_2>s_3}} W_{e^{\alpha B_{s_2} \dl}|V|} \big( (e^{\alpha B_{s_2} \dl}V) e^{i s_2 \Delta} Z(0), \\
e^{i (s_2-s_3) \Delta} (e^{\alpha B_{s_3} \dl} V) e^{i s_3 \Delta} Z(0) \big) \dd s_2 \dd s_3
\end{aligned}$$
Note that
$$
W_{e^{\alpha B_{s_2} \dl}|V|}(f, g) = W_{|V|}(e^{-\alpha B_{s_2} \dl} f, e^{-\alpha B_{s_2} \dl} g).
$$
The expression becomes
$$\begin{aligned}
\set E \int_{\substack{s_2>s_3}} W_{|V|} \big( V e^{-\alpha B_{s_2} \dl} e^{i s_2 \Delta} Z(0), \\
e^{-\alpha B_{s_2} \dl} e^{i (s_2-s_3) \Delta} (e^{\alpha B_{s_3} \dl} V) e^{i s_3 \Delta} Z(0) \big) \dd s_2 \dd s_3.
\end{aligned}$$

Next, we consider a generic term of the sum, for which we can proceed by recurrence (see Yajima \cite{yajima}). For a sesquilinear form $W$ of kernel $\widehat W(\xi_1, \xi_2)$, construct $L_V(W)$ by
$$
L_V(W)(f, g) = \set E \int_0^{\infty} W \big(V e^{-\alpha B_t \dl} e^{it\Delta} f, e^{-\alpha B_t \dl} e^{it\Delta} g \big) \dd t.
$$
Our previous construction of $W_{|V|}$ also fits this definition, if we take the preexisting form to be $\langle f, g \rangle$.

Clearly,
$$\begin{aligned}
L_V(W)(f, g) &= \set E \int_0^{\infty} \bigg(\int \widehat W(\xi_1+\eta, \xi_2) \widehat V(\eta) e^{-i\alpha B_t \xi_1 + it|\xi_1|^2} \widehat f(\xi_1) \\
&e^{i\alpha B_t \xi_2 - it|\xi_2|^2} \widehat {\ov g}(\xi_2)  \dd \eta \dd \xi_1 \dd \xi_2 \bigg) \dd t \\
&= \int_0^{\infty} \bigg(\int \widehat W(\xi_1+\eta, \xi_2) \widehat V(\eta) e^{-t \alpha^2 |\xi_1 - \xi_2|^2} e^{it|\xi_1|^2} \widehat f(\xi_1) \\
&e^{-it|\xi_2|^2} \widehat {\ov g}(\xi_2)  \dd \eta \dd \xi_1 \dd \xi_2 \bigg) \dd t
\end{aligned}$$
and then
\be\lb{2.18}\begin{aligned}
\widehat {L_V(W)}(\xi_1, \xi_2) &= \int \frac {\widehat W(\xi_1+\eta, \xi_2) \widehat V(\eta)}{\alpha^2 |\xi_1-\xi_2|^2 - i (|\xi_1|^2 - |\xi_2|^2)} \dd \eta.
\end{aligned}\ee
We retrieve exactly (\ref{2.14}) by plugging $\delta_{\xi_1 = \xi_2}$ into (\ref{2.18}).

Using this operator, we rewrite (\ref{2.12}) as
$$
\set E \int L_V(W_{|V|}) \big( e^{-\alpha B_{s_3} \dl} e^{i s_3 \Delta} Z(0), V e^{-\alpha B_{s_3} \dl} e^{i s_3 \Delta} Z(0) \big) \dd s_3.
$$

Finally, note that the expression (\ref{2.12}) can be entirely constructed by means of $L_V$ and one more operator, in which the inner multiplication of the sesquilinear form $W$ by $V$ is on the right instead of on the left:
$$
R_V(W)(f, g) = \set E \int_0^{\infty} W \big(e^{-\alpha B_t \dl} e^{it\Delta} f, V e^{-\alpha B_t \dl} e^{it\Delta} g \big) \dd t.
$$
On the Fourier side, this translates into
$$
\widehat {R_V(W)}(\xi_1, \xi_2) = \int \frac {\widehat W(\xi_1, \xi_2 - \eta) \widehat {\ov V}(\eta)}{\alpha^2 |\xi_1-\xi_2|^2 - i (|\xi_1|^2 - |\xi_2|^2)} \dd \eta.
$$
In particular, (\ref{2.12}) is given by
$$
(\ref{2.12}) = R_VL_V(W)(Z(0), Z(0)).
$$
As noted above, given that the identity $I := I(f, g) = \langle f, g \rangle$ has the symbol $\delta_{\xi_1 = \xi_2}$, for a real potential $|V|$ $W_{|V|}=R_{|V|}I = L_{|V}I$.

The expression for the other term of (\ref{2.11}), corresponding to $s_2<s_3$, is analogous, but the $L_V$ and $R_V$ operators appear in reverse order. Thus
$$
(\ref{2.11}) = R_VL_V(W_{|V|})(Z(0), Z(0)) + L_VR_V(W_{|V|})(Z(0), Z(0)).
$$

This construction easily generalizes to every term in the expansion (\ref{2.2}), except for the first one, which corresponds to the free evolution. Namely, the $n$-th term is dominated by the sum over all words of length $2n-4$ in which there are $n-2$ each of $R_V$ and $L_V$ in arbitrary order, applied to $W_{|V|}$. Each such suitable word corresponds to one possible ordering of the integration variables. This makes for exactly $\begin{pmatrix} 2n-4 \\ n-2 \end{pmatrix} < 2^{2n}$ summands.

Next, we consider the kernel of a generic summand obtained by successively applying the operators $L_{V_k}$ and $R_{V_k}$ in an arbitrary order $n$ times. It has the explicit form, which can be verified recursively,
$$
K(f, g) = \int_{(\set R^3)^n} \frac {\widehat f(\xi_0) \ov{\widehat g(\xi_n)} \prod_{k=1}^n \widehat {V_k}(\xi_{k} - \xi_{k-1}) \dd \xi_1 \ldots \dd \xi_n} {\prod_{k=1}^n \big(\alpha^2 |\xi_{a_{k}} - \xi_{b_k}|^2 - i(|\xi_{a_k}|^2 - |\xi_{b_k}|^2) \big)}.
$$
Here $[a_1, b_1] \subset [a_2, b_2] \ldots \subset [a_n, b_n] = [0, n]$ and the length increases by one at each step. Each $V_k$ can be either $V$, $\ov V$, or $|V|$.

This is entirely analogous to, but more general than, Yajima's formula, in which one always has $a_k = 0$, $b_k = k$.

However, let
$$
\sigma_k = \xi_{a_{k}} - \xi_{b_k}.
$$
After reindexing $V_k$ we obtain
$$
K(f, g) = \int_{(\set R^3)^n} \frac {\widehat f(\xi) \ov{\widehat g(\xi-\sigma_n)} \prod_{k=1}^n \widehat V_k(\sigma_k - \sigma_{k-1}) \dd \sigma_1 \ldots \dd \sigma_n} {\prod_{k=1}^n \big(\alpha^2 |\sigma_k|^2 -i|\sigma_k|^2 + 2i\sigma_k \xi_{a_k} \big)}.
$$
Since we are only interested in the $L^2$ norm, we can safely ignore the more complicated aspects of this approach. It suffices to show that, uniformly over $\xi_{a_k}$,
$$
\int_{\set R^3} \bigg| \int_{(\set R^3)^{n-1}} \frac {\prod_{k=1}^n \widehat V_k(\sigma_k - \sigma_{k-1}) \dd \sigma_1 \ldots \dd \sigma_{n-1}} {\prod_{k=1}^n \big(\alpha^2 |\sigma_k|^2 -i|\sigma_k|^2 + 2i\sigma_k \xi_{a_k} \big)}\bigg| \dd \sigma_n < \infty.
$$
Evaluating the expression in absolute value, the conclusion will follow if
$$
\int_{(\set R^3)^{n}} \frac {\prod_{k=1}^n |\widehat V_k(\sigma_k - \sigma_{k-1})| \dd \sigma_1 \ldots \dd \sigma_n} {\prod_{k=1}^n \alpha^2 |\sigma_k|^2} < \infty.
$$
This expression is a combination of convolutions and multiplications and admits a bound of
$$
\|K(f, g)\| \leq C^n \alpha^{-2n} \|V\|_{L^{3/2, 1}}^n \|f\|_2 \|g\|_2.
$$

Each term in the expansion (\ref{2.2}) is the sum of at most $2^{2n}$ such expressions. Consequently, the whole sum is bounded by
$$
\set E \||V|(x-B_t)^{1/2} Z\|_{L^2_{t, x}}^2 \leq \sum_{n \geq 1} (4C)^n \alpha^{-4n} \|V\|_{L^{3/2, 1}}^{2n} \|Z(0)\|_2^2.
$$
This converges for $\alpha$ sufficiently large relative to the norm of $V$: $\alpha \geq~C \|V\|_{L^{3/2, 1}}^{1/2}$.

In this case, the Strichartz norm of solutions corresponding to a random potential and fixed initial data will almost always be finite. Furthermore, the contribution of the potential will be, on average, in the order of $C\alpha^{-2}$:
$$
\set E \|Z - e^{it\Delta} Z(0)\|_{L^{\infty}_t L^2_x \cap L^2_t L^6_x} \leq C \alpha^{-2} \|V\|_{L^{3/2, 1}} \|Z(0)\|_2.
$$
This is obtained by starting with the second term in the summation.
\end{proof}

In particular, endpoint Strichartz estimates ensure the existence of the $L^2$ wave operators, for fixed initial data, in the following sense:
\begin{corollary}
Consider a solution $Z$ of (\ref{1.1}), with fixed initial data $Z_0$ and a potential $V \in L^{3/2, 1}$ in a state of Brownian motion:
$$
i \partial_t Z + H(t) Z = 0,\ Z(0) = Z_0,\ H(t) = -\Delta + V(x-\alpha B_t)
$$
Then for any $\alpha \geq \alpha_0$ the following limit
$$
W(Z_0) = \lim_{t \to \infty} e^{-it\Delta} Z(t)
$$
exists almost surely and, for fixed $Z_0$, 
$$
\set E \|W(Z_0) - Z_0\| \leq C \alpha^{-2} \|V\|_{L^{3/2, 1}} \|Z_0\|_2.
$$
\end{corollary}
\begin{proof} This result is obtained by applying the endpoint Strichartz estimates to equation (\ref{1.1}).
\end{proof}

\appendix
\section{Brownian motion}
As a reminder, the standard construction of three-dimensional Brownian motion is as follows. There exist a probability space $(\Omega, \set P(\omega))$ and random variables $B_t: \Omega \to \set R^3$, $t \geq 0$, such that
\begin{enumerate}
\item[i)] For every $0<t<s$, $B_t-B_s$ and $B_s$ are independent random variables.
\item[ii)] For every $t>0$, $B_t$ has a normal distribution with mean $0$ and variance $t$, $N(0, t)$, of probability density
\be
p(B_t = y) = (2\pi t)^{-3/2} e^{-|y|^2/2t}.
\ee
\item[iii)] Almost surely $B_t$ depends continuously on $t$.
\end{enumerate}
Such a family of random variables $B_t$ is called a \emph{Brownian motion} and is uniquely characterized, up to a measure-preserving transformation, by these three properties.

Brownian motion obeys the scaling $\alpha B_{t} = \tilde B_{\alpha^2 t}$, where $\tilde B_t$ is a distinct Brownian motion.

Locally in time, Brownian motion belongs almost surely to the H\"{o}lder spaces $\Lambda^s$, $s<1/2$, and to the Besov spaces $B^{1/2}_{p, \infty}$, $1\leq p<\infty$.

\section*{Acknowledgments}
M.B.\ is partially supported by a Rutgers Research Council grant.

A.S.\ is partially supported by the NSF grant DMS--0903651.



\begin{thebibliography}{AbcDef1}
\bibitem[Agm]{agmon} S.\ Agmon, \emph{Spectral properties of Schr\"{o}dinger operators and scattering theory}, Ann.\ Scuola Norm.\ Sup.\ Pisa Cl.\ Sci.\ (4)  2 (1975), No.\ 2, pp.\ 151--218.

\bibitem[AmGe]{amge} W. Amrein and V. Georgescu, \emph{On the characterization of bounded states and scattering states in quantum mechanics}, Helv.\ Phys.\ Acta 46 (1973), pp.\ 635--658.


\bibitem[BeSo]{becsof} M.\ Beceanu, A.\ Soffer, \emph{The Schr\"{o}dinger  equation with a potential in rough motion}, to appear in Communications in PDE, arXiv:1103.0521.

\bibitem[BeL\"o]{bergh} J.\ Bergh, J.\ L\"ofstr\"om, \emph{Interpolation Spaces. An Introduction}, Springer-Verlag, 1976.

\bibitem[Che1]{cherem} S.\ E.\ Cheremshantsev, \emph{Asymptotic completeness and the absence of bound states in the quantum problem of scattering on a Brownian particle}, Mathematical Notes, Vol.\ 46, No.\ 4,  1989.

\bibitem[Che2]{cherem2} S.\ E.\ Cheremshantsev, \emph{Theory of scattering by a Brownian particle}, Collection of articles, Trudy Mat.\ Inst.\ Steklov., 184, Nauka, Leningrad, 1990, pp.\ 5 --104.

\bibitem[Ens]{ens} V.\ Enss, \emph{Asymptotic completeness for quantum mechanical potential scattering, i.\ short range potentials}, Commun.\ Math.\ Phys.\ 61 (1978), pp.\ 285--291.




\bibitem[KeTa]{tao} M.\ Keel, T.\ Tao, \emph{Endpoint Strichartz estimates}, Amer.\ Math.\ J.\ 120 (1998), pp.\ 955--980.

\bibitem[MPSW]{brown} P.\ M\"{o}rters, Y.\  Peres, O.\ Schramm, W.\ Werner, \emph{Brownian Motion}, Cambridge University Press, 2010.


\bibitem[Pil]{pillet} C.\ A.\ Pillet, \emph{Asymptotic completeness for a quantum particle in a Markovian short range potential}, 	Communications in Mathematical Physics, Vol.\ 105, No.\ 2, 1986, pp. 259--280.


\bibitem[Rue]{ruelle} D.\ Ruelle, \emph{A remark on bounded states in potential scattering theory}, Nuovo Cimento 61A (1969), pp.\ 655--662.


\bibitem[Ste]{stein} E.\ Stein, \emph{Harmonic Analysis}, Princeton University Press, Princeton, 1994.

\bibitem[Tao]{taobook} T.\ Tao, \emph{Nonlinear Dispersive Equations: Local and Global Analysis}, CBMS Regional Conference Series in Mathematics, 106. Published for the Conference Board of the Mathematical Sciences, Washington, DC; by the American Mathematical Society, Providence, RI, 2006.

\bibitem[Tay]{taylor} M.\ E.\ Taylor, \emph{Tools for PDE.\ Pseudodifferential operators, paradifferential operators, and layer potentials}, Mathematical Surveys and Monographs, 81, American Mathematical Society, Providence, RI, 2000.

\bibitem[Yaj]{yajima} K.\ Yajima, \emph{The $W^{k, p}$-continuity of wave operators for Schr\"{o}dinger operators}, J.\ Math.\ Soc.\ Japan 47 (1995), pp.\ 551--581.

\end{thebibliography}
\end{document}